\documentclass[11pt]{article}
\usepackage[margin=1in]{geometry}

\usepackage{amsmath,amsbsy}
\usepackage{amsthm} %
\usepackage{amsfonts}
\usepackage{setspace}
\usepackage{color}
\usepackage{xcolor}
\usepackage{algorithm}
\usepackage[noend]{algpseudocode}
\usepackage{multirow}
\usepackage{paralist}
\usepackage{xspace}
\usepackage{threeparttable}
\usepackage{booktabs}
\usepackage{caption}
\usepackage{subcaption}
\definecolor{MyGreen}{rgb}{0.1333,0.5451,0.1333}
\usepackage[linktocpage=true, 
	pagebackref=true,colorlinks,
  linkcolor=blue, urlcolor=blue,
	citecolor=MyGreen,
	bookmarks,bookmarksopen,bookmarksnumbered]
	{hyperref}
\let\originalleft\left
\let\originalright\right
\renewcommand{\left}{\mathopen{}\mathclose\bgroup\originalleft}
\renewcommand{\right}{\aftergroup\egroup\originalright}

\renewcommand{\Pr}{\operatorname*{\textbf{\textup{Pr}}}}

\DeclareMathOperator*{\dcup}{\mathbin{\dot\cup}}

\DeclareMathOperator*{\EE}{\textbf{\textup{E}}}

\renewcommand{\epsilon}{\varepsilon}
\newcommand{\set}[1]{\{ #1 \}}
\newcommand{\Set}[1]{\left\{ #1 \right\}}

\newcommand{\lb}{\left}
\newcommand{\rb}{\right}
\newcommand{\lt}{\left}
\newcommand{\rt}{\right}
\newcommand{\md}{\middle}

\newcommand{\range}{\textsf{range}}

\DeclareMathOperator{\poly}{\ensuremath{\mathrm{poly}}}

\makeatletter
\newtheorem*{rep@theorem}{\rep@title}
\newcommand{\newreptheorem}[2]{%
\newenvironment{rep#1}[1]{%
\def\rep@title{#2 \ref{##1}}%
\begin{rep@theorem}[restated]}%
{\end{rep@theorem}}}
\makeatother
\newreptheorem{lemma}{Lemma}
\newreptheorem{theorem}{Theorem}
\newreptheorem{corollary}{Corollary}

\newboolean{short}
\newcommand{\onlyShort}[1]{\ifthenelse{\boolean{short}}{#1}{}}
\newcommand{\onlyLong}[1]{\ifthenelse{\boolean{short}}{}{#1}}

\theoremstyle{definition}

\newtheorem{definition}{Definition}

\theoremstyle{plain}
\newtheorem{lemma}{Lemma}
\newtheorem{theorem}{Theorem}
\newtheorem{corollary}{Corollary}
\newtheorem{open problem}{Open Problem}

\usepackage[most]{tcolorbox}
\tcbuselibrary{theorems}

\newcommand{\ann}[1]{%
\text{\footnotesize(#1)}\quad}

\usepackage{thmtools}
\usepackage{mathtools}

\title{Polynomial Lower Bounds for Distributed Graph Sketching with Tiny Error: Connectivity and Spanning Tree Construction} 

\date{}

\author{Peter Robinson\thanks{Peter Robinson
was supported in part by National Science Foundation (NSF) grants CCF-2402836 and CCF-2552881.}\\
\small{Computer \& Cyber Sciences}\\
\small{Augusta University}
\and
Ming Ming Tan\thanks{
Ming Ming Tan was supported in part by National Science Foundation (NSF)
grant CCF-2348346 CRII.} \\
\small{Computer \& Cyber Sciences}\\
\small{Augusta University}
}

\setboolean{short}{false}

\newcommand{\fail}{\ensuremath{\textsf{Fail}}}
\newcommand{\allunsafe}{\ensuremath{\textsf{All-Unsafe}}}
\newcommand{\allbad}{\ensuremath{\textsf{All-Dangerous}}}
\newcommand{\Gconn}{\ensuremath{G}_{X,J}^{\text{conn}}}
\newcommand{\GST}{\ensuremath{G}_{X}^{\text{ST}}}
\newcommand{\Dconn}{\ensuremath{\mathcal{D}}_{n,q}^{\text{conn}}}
\newcommand{\DST}{\ensuremath{\mathcal{D}}_{n,q}^{\text{ST}}}

\begin{document}

\maketitle
\begin{abstract}
We present the first polynomial lower bounds for several fundamental problems in the distributed graph sketching model in the tiny-error regime, which includes deterministic algorithms as a special case. In the graph sketching model, every node sends a single message to the referee who does not have any prior knowledge of the graph and must output the answer. 
While the work of Nelson and Yu (SODA 2019) and Yu (SODA 2021) showed that $\Theta\lt( \log^3n \rt)$ is optimal for constructing a spanning forest or deciding whether the graph is connected with error at most $\frac{1}{\poly(n)}$ , their approach does not yield any stronger bounds for significantly smaller error probabilities.

Our main result is to show that solving either connectivity or spanning tree construction with error at most $\delta$ requires messages of length $\Omega\lt( \min\set{n, \log_2 \tfrac{1}{\delta}}^{1/3} \rt)$, which implies that algorithms with exponentially small error must send messages of $\Omega\lt( n^{1/3} \rt)$ bits in the worst case.
Our results significantly narrow the current gap between the Nelson-Yu threshold of $\Theta\lt( \log^3n \rt)$ and the trivial upper bound of sending $O(n)$ bits per node for deterministic graph sketching.

We also extend our results to $k$-edge connectivity. 
For any $k=O(n^{1/7})$, we recover the same bound of $\Omega\lt( k \rt)$ on the message length for algorithms with exponentially small error that was shown by Robinson and Tan (PODS 2026) only for deterministic algorithms.
Finally, for $k=n^{o(1)}$, our result implies a stronger lower bound of $\Omega_\epsilon\lt( n^{\epsilon} \rt)$ bits, for any constant $\epsilon<\tfrac{1}{3}$.
\end{abstract}

\section{Introduction} \label{sec:intro}

In the distributed graph sketching model~\cite{becker2011adding}, we consider a graph of $n$ nodes that are equipped with unique integer IDs and a referee. 
The input of each node $u$ consists of the list of its neighbors' IDs in the graph, and $u$ may send a single message to the referee, who produces the output depending on the received messages. 
The naive approach for solving any problem in the sketching model is to simply instruct every node to send a message containing all its neighbors' IDs to the referee, which, however, requires a linear message length and does not scale to large graphs. 
Thus, the main challenge is designing algorithms that are \emph{efficient}, in the sense that each node's message is limited to a polylogarithmic number of bits, i.e., nodes only send a \emph{sketch} of their neighborhoods to the referee.

Deciding whether the graph is connected or constructing a spanning tree are arguably some of the most fundamental graph problems. 
The work of Ahn, Guha, and McGregor~\cite{AGM-soda12} shows that these problems have low-memory implementations in the semi-streaming setting, and it is straightforward to adapt their technique to yield distributed graph sketching algorithms that correctly solve these problems with probability at least $1 -\frac{1}{\poly(n)}$, while sending messages of $O\lt( \log^3n \rt)$ bits per node. 
The optimality of this threshold remained open for several years, until Nelson and Yu~\cite{NY19-soda} showed that $\Theta\lt( \log^3n \rt)$ is indeed tight for spanning forest construction; subsequently, Yu~\cite{Y-soda21} proved that the same lower bound holds even for connectivity testing.
A natural question to ask is whether access to randomness and allowing a small probability of error are indeed necessary for obtaining efficient graph sketching algorithms with small-length messages. 
To date, there are no deterministic upper bounds known for connectivity in the distributed graph sketching model that would ensure a worst case message size of $o(n)$, which suggests that we should not hope for a deterministic connectivity algorithm with polylogarithmic message-size. 
Interestingly, the aforementioned lower bounds of \cite{NY19-soda,Y-soda21} do not suffice to resolve this question, as their simulation argument introduces an additive error term of roughly $1/n^{\epsilon}$, for some constant $\epsilon>0$, which prevents these techniques from yielding stronger bounds for algorithms with a much smaller probability of error.  (We elaborate in more detail on this point in Section~\ref{sec:overview}.)  
Consequently, the best known bounds in the existing literature leave a large gap between the $\Omega\lt(  \log^3 n\rt)$ barrier and the trivial upper bound of $O(n)$ bits for deterministic connectivity and spanning tree algorithms in the distributed graph sketching model.
Obtaining lower bounds on the deterministic complexity of connectivity has been of significant interest to the graph sketching community.
For instance, \cite{assadi2022lower} (p.\ 102) mentions 
\begin{quote}
\sl
``Prove an $n^{\Omega\lt( 1 \rt)}$ lower bound  for deterministic connectivity''
\end{quote}
under ``\textsl{Harder Open Problems}'' as a ``\textsl{longstanding open question}'', and \cite{DBLP:journals/corr/abs-2510-16336} state 
\begin{quote}
\sl
``Establishing lower bounds in other settings---distributed or streaming,
deterministic or randomized---remains an intriguing open question.''
\end{quote}
Our work takes a significant step towards resolving this fundamental question.

\subsection{Preliminaries: Sketching Model and Graph Problems} \label{sec:prelim}
We consider $n$ nodes each of which is equipped with a unique integer ID of $\Theta(\log n)$ bits from $[n]$.
Every node knows its neighbors' IDs as well as $n$, and may send a single message of $B$ bits to the referee, who does not have any prior information about the graph and who computes the answer based on the received messages.
For randomized algorithms, we assume that the nodes and the referee also have shared access to an infinite string of random bits. 

For \emph{spanning tree construction}, the referee needs to output the list of edges that form a spanning tree of the graph, whereas for \emph{connectivity}, the referee simply outputs a bit to indicate whether the graph is connected.
Finally, \emph{$k$-edge connectivity} requires the referee to answer whether every nontrivial cut of the graph contains at least $k$ edges.

\subsection{Our Contributions} \label{sec:contrib}

We present the first polynomial lower bounds for connectivity testing and spanning tree construction in the \emph{tiny error} regime, i.e., algorithms that fail with probability at most $\frac{1}{n^{\poly\log(n)}}$, which includes deterministic algorithms as an important special case.
Our main result is the following: 

\newcommand{\thmMain}{
Any randomized algorithm that decides connectivity  or computes a spanning tree on an $n$-node graph in the distributed graph sketching model with error at most $\delta\le\frac{1}{n^{\log^8n}}$, has a worst case message length of
\begin{align}
\Omega\lt( \min\Set{n, \log_2 \frac{1}{\delta}}^{1/3} \rt).
\notag 
\end{align}
} 

\begin{theorem} \label{thm:main}
\thmMain
\end{theorem}

Instantiating Theorem~\ref{thm:main} with $\delta=2^{-n}$ immediately gives the following:

\begin{corollary} \label{cor:det}
Any deterministic algorithm for connectivity testing or spanning tree construction requires a worst case message length of $\Omega\lt( n^{1/3} \rt)$ bits in the distributed graph sketching model.
The same bound holds for randomized algorithms with exponentially small error.
\end{corollary}
At first glance, it may seem that the bound for the spanning tree problem follows immediately from the one for connectivity. While this is true when considering algorithms for the \emph{spanning forest} problem, which must work on all (connected or disconnected) graphs, here we explicitly consider spanning tree algorithms that are only guaranteed to work on connected graphs, which requires a separate argument. 

We point out that the proof of Theorem~\ref{thm:main} remains valid for algorithms with a significantly larger error probability than the stated one. However, the obtained bound reduces to the known $\Theta\lt(  \log^3 n\rt)$ threshold of \cite{Y-soda21,NY19-soda} as soon as $\delta=\Theta\lt( \frac{1}{n^{\log^8n}} \rt)$, and does not yield any useful bound once $\delta\ge \omega\lt( \frac{1}{n^{\log^8n}} \rt)$.

By using a standard ``blow-up'' construction, where we obtain a larger graph by replacing nodes with cliques, we can extend our connectivity result to $k$-edge connectivity:

\begin{corollary} \label{cor:kedge}
Any randomized graph sketching algorithm that decides $k$-edge connectivity on an $n$-node graph with error at most $\delta\le\frac{1}{n^{\log^8n}}$, has a worst case message length of
\begin{align}
\Omega\lt( \frac{1}{k} \cdot \min\Set{\frac{n}{k},\log_2\frac{1}{\delta}}^{1/3}  \rt).
\notag 
\end{align}
For the exponentially small error regime $\delta\le 2^{-\alpha n}$, for fixed $\alpha>0$, and for deterministic algorithms we obtain the following:
\begin{compactitem} 
\item For any $k = O\lt(   n^{1/7}\rt)$, the message length is $\Omega\lt( k \rt)$.
\item If $k=n^{o(1)}$ , then, for every fixed $\epsilon<\frac{1}{3}$, the message length is $\Omega_{\alpha,\epsilon}\lt( n^{\epsilon} \rt)$ bits.\footnote{Note that the notation $\Omega_\gamma\lt( \dots \rt)$ means that the hidden constant may depend on $\gamma$.}
\end{compactitem}
\end{corollary}
Note that Corollary~\ref{cor:kedge} recovers the bound of $\Omega\lt( k \rt)$ shown recently by Robinson and Tan~\cite{robinson2026deterministic} for deterministic algorithms.
In fact, as long as $k$ is subpolynomial, we obtain a bound of $\Omega\lt( n^{\epsilon} \rt)$, thus improving over their result.

\subsection{Technical Challenges and Our Approach} \label{sec:overview}

A major technical challenge in the sketching model emanates from the assumption that every edge is shared between two nodes.
In particular, this ``vertex partitioning'' property of the inputs prevents us from leveraging direct reductions from problems in communication complexity, which is a common lower bound technique in the setting where the edges are partitioned between the players, cf.\ \cite{woodruff2017distributed}. 
The existing lower bounds of \cite{NY19-soda,Y-soda21} circumvent this obstacle by implementing a ``lossy'' simulation of the sketching algorithm in the 2-party model of communication complexity.
The purpose of the simulation is to construct a solver for variants of the universal relation (UR) problem~\cite{karchmer1995super,kapralov2017optimal} in the one-way communication model, where Alice gets a subset $S$ of some universe and sends a single message to Bob, whose input consists of some proper subset $T \subset S$. Upon receiving Alice's message, Bob must output some element in $S  \setminus T$.\footnote{Strictly speaking,  \cite{Y-soda21} considers a decision variant of the UR problem. Here we focus our discussion on the search variant to keep the presentation simple.}
The UR problem is a classic find-the-needle-in-the-haystack problem, which makes it a natural candidate for proving a lower bound for the amount of communication required to ensure that a node can successfully identify a crucial spanning tree edge among many adjacent edges.  
 
To address the technical challenge of the input edges being shared between their neighbors, the simulation of \cite{NY19-soda,Y-soda21} introduces an error term (in addition to the error of the algorithm), by giving up on simulating the messages sent by a certain subset of nodes. Nevertheless, they show via Pinsker's inequality~\cite{pinsker1964information} that the resulting distribution observed by Bob is statistically close enough to the distribution at the referee in the sketching model.  
Since the error term of their simulation is roughly $\frac{1}{{n}^{\epsilon}}$, for some constant $\epsilon>0$, this approach does not provide better bounds for algorithms in the tiny error regime, including deterministic algorithms.

While our graph construction is similar to the ones in \cite{NY19-soda,Y-soda21}, our overall argument is very different, as we do not attempt to achieve a reduction from a communication complexity problem.
In the graph, there are $M$ center nodes, and each center $c_i$ has some ``noisy'' edges to its \emph{private neighbors}, which are not adjacent to any other center nodes. In addition, there are disjoint sets $R_0$ and $R_1$ of so-called  \emph{shared neighbors}, and $c_i$ is connected to a subset $D_i$ that lies entirely in either $R_0$ or $R_1$.
That is, for our hard distribution, we first create a partition of the ID space into disjoint sets $P_1,\dots,P_M,R_0,R_1$, all of which have the same size $q$, and then we sample the corresponding private neighbors $T_i  \subseteq P_i$ and shared neighbors $D_i$ for each center node $c_i$. 
To determine whether we choose $D_i$ such that $D_i \subseteq R_0$ or $D_i  \subseteq R_1$, we sample an $M$-length binary vector $X$.

We can think of $(P_i,R_0,R_1)$ as being an \emph{ordered balanced split} (OBS) of the possible ID space of $c_i$'s neighbors. 
Note that, after choosing the actual neighbors of $c_i$ from these subsets, $c_i$'s view consists of the IDs in $D_i \cup T_i$ without knowing which ID lies in which set.
That is, the message sent to the referee is determined by applying a function to the entire set $D_i \cup T_i$.
We call an OBS \emph{safe} if for any given set of private neighbors $T_i$, there is no way of constructing two inputs for $c_i$ on which it sends the same message to the referee, where we choose the subset of shared neighbors to lie in $R_0$ in the first input and in $R_1$ in the second input.
We then bound the $3$-ary Vapnik–Chervonenkis (VC) dimension~\cite{karpovsky1978coordinate,mossel2002complexity} of the safe OBSs of $c_i$.
This enables us to use a theorem of \cite{karpovsky1978coordinate} (also known as generalized Sauer bound) to obtain an upper bound on the size of any such family of safe OBSs.
By using a combinatorial argument on multinomial coefficients, it follows that the probability of sampling a safe OBS for a center node $c_i$ is exponentially small in $q$. 

Subsequently, assuming that indeed all center node OBSs are unsafe, we derive a lower bound on the probability of actually sampling such ``dangerous'' triples $(T_i, D_i^0, D_i^1)$ for every center $c_i$, which serve as witnesses of the assumed non-safety of the OBSs. 

We show that the correct spanning tree crucially depends on the values of the sampled binary vector $X$. 
However, conditioned on all sampled triples being dangerous, it is impossible for a center node to convey the difference to the referee, since each $c_i$ must send the same message for the graph constructed with $X_i=0$ as it does for the graph for $X_i=1$.
Of course, the shared neighbors of $c_i$ do have a different view (due to having distinct sets of center neighbors) in these two cases and could, at least in principle, alert the referee of the difference. 
However, their number is polynomially smaller than the number of center nodes. 
Thus, for the spanning tree construction problem, we can show that this imbalance prevents the algorithm from computing the correct answer unless the message length is sufficiently large. 
For the connectivity problem, we use a standard encoder-decoder information theoretic argument to obtain the result.

\subsection{Additional Related Work}  \label{sec:related}

The first work to consider the distributed graph sketching model was by Becker, Matamala, Nisse, Rapaport, Suchan, and
Todinca in~\cite{becker2011adding}, who proved the hardness of several graph problems for deterministic sketching algorithms, including diameter testing and subgraph detection.
They obtain their results by proving that the existence of an efficient protocol for these problems allows one to reconstruct the entire graph, which, of course is impossible without a sufficiently large bound on the message length. 
They explicitly point out that their approach does not extend to graph connectivity problems.

The graph sketching model was further studied by Becker, Montealegre, Rapaport, and Todinca in \cite{BMRT-sirocco14}, who show separations between the power of deterministic, private randomness and public randomness algorithms.
More recently, Assadi, Kol, and Oshman~\cite{AKO-podc20} showed a lower bound of $\Omega\lb(n^{1/2-\epsilon}\rb)$ bits on the required message length  for computing an MIS.
In the variant of the distributed sketching model where nodes have access to only \emph{private} randomness, Holm, King, Thorup, Zamir, and Zwick~\cite{HKTZZ-focs19} showed that computing a spanning tree is possible with sketches of $O(\sqrt{n}\log n)$ bits.

There is a known equivalence between the distributed graph sketching model and the single-round variant of the broadcast congested clique, as observed by \cite{jurdzinski2018communication,AKO-podc20}.
In the latter model, each one of $n$ nodes can broadcast a single message per round that is received by the other nodes at the end of the current round.
Pai and Pemmaraju~\cite{PP20-fsttcs} showed that $\Omega\lt( \frac{\log n}{b} \rt)$ rounds are required assuming that nodes can broadcast messages of length $b$ bits.
Montealegre and Todinca~\cite{MT16-podc} discovered that there is a deterministic $r$-round connectivity algorithm that sends messages of size $O(n^{1/r}\log n)$.
In subsequent work, Jurdzinski and Nowicki~\cite{JN-disc17} showed how to improve the round complexity to $O(\log n/\log \log n)$ when considering the standard assumption of $O(\log n)$ bits per message.

\subsection{Roadmap}

In Section~\ref{sec:obs}, we define ordered balanced splits (OBSs) of ID sets and prove a bound on their $3$-ary VC dimension.
We also prove an upper bound on the probability of obtaining a safe OBS.
In Section~\ref{sec:dist}, we define the hard input distribution as a sequence of four sampling steps, which is used to determine the node IDs and the adjacencies for the spanning tree and connectivity graph constructions, and state some of their crucial properties.
In Section~\ref{sec:dangerous}, we continue developing the combinatorial argument of Section~\ref{sec:obs}, by bounding the probability of sampling a so-called ``dangerous'' triple, which can be viewed as a witness for the non-safety of an OBS. Moreover, we also show how conditioning on the event of sampling only dangerous triples impacts the transcript of the shared neighbors (which we call right-side transcript).
Then, we show that the distributional error is sufficiently large for spanning tree construction in Section~\ref{sec:distST} and for connectivity testing in Section~\ref{sec:distconn}.
We complete the proof of our main result (Theorem~\ref{thm:main}) for randomized algorithms via Yao's lemma in Section~\ref{sec:random}.
Finally, we extend the results to $k$-edge connectivity in Section~\ref{sec:kedge}.
 
\section{Ordered Balanced Splits} \label{sec:obs}

Throughout this section, we consider two integer parameters $q$ and $B$. We assume that $q$ is sufficiently large and that 
\begin{align}
B \le  \frac{q}{100}. \label{eq:B}
\end{align}
Moreover, we consider $U$ to be a set of size $3q$ and a function $f : 2^U \to \mathcal{M}$, where $|\mathcal{M}| \le 2^B$.
Intuitively speaking, we can think of $U$ as a set of possible neighborhood IDs and $f$ as a node's function that maps a given subset of neighbors to a message from some alphabet $\mathcal{M}$.

\begin{definition} \label{def:obs}
An \emph{ordered balanced split (OBS) of $U$} is a partition of $U$ into subsets $(P,R_0,R_1)$, where $|P|=|R_0|=|R_1|=q$. 
\end{definition}

Note that we can think of an OBS of $U$ as a 3-coloring of $U$ with color palette $\textsf{Col}=\set{\textsf{P},\textsf{R}_0,\textsf{R}_1}$ (i.e., a map $U \to \textsf{Col}$), if we restrict all colors to occur equally often.
The following is an adaptation of Definition~2.1 of \cite{mossel2002complexity}:

\begin{definition}[Shattering and $3$-ary VC dimension] \label{def:shatter}
Let $\mathcal{H}$ be a family of ordered balanced splits of $U$,
i.e., $\mathcal{H}  \subseteq\set{ \phi: U \to \textsf{Col}}$.
We say that \emph{$\mathcal{H}$ shatters $Y$} if every possible $3$-coloring of $Y$ can be obtained by restricting the domain of some $\phi \in \mathcal{H}$ to $Y$.
The \emph{$3$-ary VC dimension of $\mathcal{H}$} is the size of the largest set that is shattered by $\mathcal{H}$. 
\end{definition}

The hard lower bound graph instances that we construct in Section~\ref{sec:dist} will leverage the fact that the referee is unable to compute the correct answer unless nodes distinguish certain input pairs by sending distinct messages.
This is captured by the following notion of ``safe'' behavior:

\begin{definition} \label{def:safe}
We say that an OBS $(P,R_0,R_1)$ is \emph{safe for $f$} if, for every $T  \subseteq P$, and every pair of nonempty sets $D_0 \subseteq R_0$ and $D_1 \subseteq R_1$, it holds that
\begin{align}
	f(T \cup D_0) \ne f(T \cup D_1).
\end{align}
\end{definition}

Let $\mathcal{H}_f$ be the family of all safe OBS for the given function $f$.

\begin{lemma} \label{lem:shatters_safe}
If $\mathcal{H}_f$ shatters a set $Y  \subseteq U$, then, for every pair of distinct subsets $A,A'  \subseteq Y$ of size $\lt\lfloor |Y| / 2 \rt\rfloor$, we have $f(A) \ne f(A')$.
\end{lemma}
\begin{proof}
Assume towards a contradiction that the statement is false, which means that there are $A, A' \subseteq Y$ such that $f(A) = f(A')$.
Consider a $3$-coloring $\phi$ of $Y$ such that 
\begin{align}
\phi(y) = 
\begin{cases}
{\textsf{R}_0} & y \in A  \setminus A'; \\
\textsf{R}_1 	   & y \in A'  \setminus A; \\
\textsf{P} 	& y \in A \cap A'; \\
\textsf{P} 	& y \in Y \setminus (A \cup A').\\
\end{cases}%
\end{align}
Since $\mathcal{H}_f$ shatters $Y$, there exists an OBS $(P,R_0,R_1)$ that extends $\phi$ to a (full) coloring of $U$.
Choose $T=A \cap A'$, $D_0=A  \setminus A'$, and $D_1=A'  \setminus A.$
By assumption, we get  
\begin{align}
	f(T \cup D_0) = f(A) = f(A') = f(T \cup D_1),
\end{align}
which shows that $(P,R_0,R_1)$ is not safe, thus providing a contradiction. 
\end{proof}

\begin{lemma} \label{lem:dim}
Recall that the range of function $f$ is of size $2^B$.
It holds that the $3$-ary VC dimension of $\mathcal{H}_f$ is at most $2B+3$.
\end{lemma}
\begin{proof}
Consider any $Y \subseteq U$ of size $d$ that is shattered by $\mathcal{H}_f$.
Recall from Lemma~\ref{lem:shatters_safe} that $f$ yields pairwise distinct values for all subsets of $Y$ that have size $\lt\lfloor d/2 \rt\rfloor$, which implies that
\begin{align}
	\binom{d}{\lt\lfloor d/2 \rt\rfloor} \le 2^B. \label{eq:B1}
\end{align}
In the remainder of the proof, we obtain an upper bound of $d$ in terms of $B$.
Since $\binom{d}{\lt\lfloor d/2 \rt\rfloor}$ is the largest binomial coefficient in the sum 
$\sum_{t=0}^{d}{d \choose t} = 2^d$ and
considering that there are $(d+1)$ terms in the sum, it follows that
\begin{align}
\binom{d}{\lt\lfloor d/2 \rt\rfloor} &\ge \frac{2^d}{d+1}.\notag 
\end{align} 
Note that the right-hand side is an increasing function in $d$.
Thus, if it was true that $d \ge 2B+4$, we would get
\begin{align}
\frac{2^d}{d+1} \ge \frac{2^{2B+4}}{2B+5} > 2^B,
\end{align}
contradicting \eqref{eq:B1}.
We conclude that $d \le 2B+3$.
\end{proof}

In the proof of Lemma~\ref{lem:Hsize} below, we make use of the following result of \cite{karpovsky1978coordinate} that we restate for completeness:
\begin{theorem}[\cite{karpovsky1978coordinate}] \label{thm:sauer}
Let $\mathcal{F}$ be a family of $3$-colorings of an $N$-element set and suppose that the $3$-ary VC dimension of $\mathcal{F}$ is at most $d$.
Then $|\mathcal{F}| \le \sum_{t=0}^{d} \binom{N}{t} 2^{N-t}$.
\end{theorem}

\begin{lemma} \label{lem:Hsize}
It holds that $|\mathcal{H}_f| \le 2^{3q}(q+1)(192e)^{q/32}$.
\end{lemma}
\begin{proof}
Let $c_0=\lt\lfloor q/32 \rt\rfloor$.
Combining Theorem~\ref{thm:sauer} with the upper bound on the dimension derived in Lemma~\ref{lem:dim} yields
\begin{align}
|\mathcal{H}_f| 
\le \sum_{t=0}^{2B+3} {3q \choose t} 2^{3q-t} 
\le  2^{3q} \sum_{t=0}^{2B+3} {3q \choose t}  
\le  2^{3q} \sum_{t=0}^{c_0} {3q \choose t},
\end{align} 
where we have used the fact that $2B+3 \le c_0$, for sufficiently large $q$ in the last inequality.
Since the binomial coefficients ${3q \choose t}$ are increasing for $0 \le t \le c_0$, it follows that
\begin{align}
|\mathcal{H}_f| 
\le 2^{3q} \sum_{t=0}^{c_0} {3q \choose c_0}
&= 2^{3q} (c_0+1) {3q \choose c_0} \notag\\ 
\ann{since $c_0 \le q$}
&\le 2^{3q} (q+1) {3q \choose c_0} \notag\\ 
&\le 2^{3q} (q+1) \lt(\frac{3eq}{c_0}\rt)^{c_0}. \notag\\ 
\end{align} 
As long as $q\ge 64$, it holds that $c_0 \ge \frac{q}{64}$, and thus
\begin{align}
|\mathcal{H}_f|
&\le 2^{3q}(q+1) \lt(192e\rt)^{c_0} 
\le
2^{3q}(q+1)(192e)^{q/32},
  \notag
\end{align}
where the last inequality follows from $c_0\le q/32$.
\end{proof}

\begin{lemma} \label{lem:OBS_safe}
Let $(P,R_0,R_1)$ be a uniformly at random sampled OBS of $U$.
It holds that $(P,R_0,R_1)$ is safe with probability at most $2^{-q}$.
\end{lemma}
\begin{proof}
Since Lemma~\ref{lem:Hsize} already provides an upper bound on the size of $\mathcal{H}_f$, it suffices to show that the total number of OBSs of $U$ is sufficiently large.
We start by observing that the total number of OBSs corresponds to the multinomial
\[
{3q \choose q,q,q},
\]
which we will lower-bound next.
The Multinomial Theorem (e.g., see Sec.~6.5 in \cite{rosen2019discrete}) tells us that
\begin{align}
 \sum_{a+b+c=3q}{3q \choose a,b,c} = 3^{3q},  \label{eq:sum1}
\end{align}
where the left-hand side is the sum over the multinomial coefficients with three nonnegative parts that add up to $3q$.

By a straightforward application of the ``stars and bars'' counting method (see, e.g., Chap.~6 in \cite{rosen2019discrete}), the total number of triples $(a,b,c)$ with the property that $a+b+c=3q$ is 
\begin{align}
{3q+2 \choose 2} = \frac{(3q+1)(3q+2)}{2} < (3q+1)^2,
\label{eq:triples}
\end{align}
where the inequality follows from the fact that $3q+2<2(3q+1)$.
 
Consider two multinomial coefficients ${3q \choose a,b,c}$ and ${3q \choose a-1,b+1,c}$, and suppose that two coordinates differ by at least $2$; w.l.o.g., $a \ge b+2$.
Since 
\begin{align}
\frac{{3q \choose a-1,b+1,c}}{{3q \choose a,b,c}} = \frac{a}{b+1} > 1,
\end{align}
it follows that making a multinomial more balanced increases its size.
Thus, ${3q \choose q,q,q}$ is the largest term in $\sum_{a+b+c=3q}{3q \choose a,b,c}$, and hence is at least as large as the average. 
Recalling \eqref{eq:sum1} and \eqref{eq:triples}, this implies 
\begin{align}
{3q \choose q,q,q}
\ge
\frac{3^{3q}}{(3q+1)^2}.
\end{align}
It follows that
\begin{align}
\frac{|\mathcal{H}_f|}{{3q \choose q,q,q}} 
&\le (q+1)(3q+1)^2 \gamma^q,
\end{align}
where $\gamma = \frac{8(192e)^{1/32}}{27} < \frac{1}{2}$.
In particular, there exists some $\eta>0$ such that $\gamma = 2^{-1 - 2\eta}$.
Since $(q+1)(3q+1)^2 = 2^{o(q)}$, for sufficiently large $q$, it follows that $(q+1)(3q+1)^2 \le 2^{\eta q}$.
We conclude that 
\begin{align}
\frac{|\mathcal{H}_f|}{{3q \choose q,q,q}} 
\le 2^{-(1+\eta)q} \le 2^{-q}.\notag 
\end{align}
\end{proof}

\section{The Hard Input Graph Distribution} \label{sec:dist}

We now describe how we sample the lower bound graph.
Note that our graph construction is similar to the constructions pioneered by \cite{Y-soda21,NY19-soda}, and further explored in \cite{robinson2023distributed}.

For the given parameter $q$, let 
\begin{align}
	M = q^2. \label{eq:M}
\end{align}
We describe the distribution as four separate sampling steps, as we will need to condition on the information revealed in the individual sub-steps below.
\begin{itemize} 
\item First, we assign the integers in the set $W = [(M+2)q]$ into randomly sampled subsets as follows:
    \begin{itemize} 
    \item \textbf{Step 1A:} Choose an ordered partition $P_1 \dcup \cdots \dcup P_M \dcup R_0 \dcup R_1$ of $W$ uniformly at random, under the restriction that all of these subsets have size $q$.\footnote{Note that $\dcup$ denotes the disjoint union operator.} We say that $U_i=P_i \cup R_0 \cup R_1$ is the \emph{support of center $c_i$}.
    \item \textbf{Step 1B:} Then, for every $i \in [M]$, independently and uniformly choose
    $T_i \subseteq P_i$ and nonempty sets $D_i^0  \subseteq R_0$ and $D_i^1  \subseteq R_1$.
    \end{itemize} 
\item \textbf{Step 2:}
Sample independently and uniformly a binary vector $X = (X_1,\dots,X_M)$.
\item \textbf{Step 3:}
Sample an independent and uniform index $J \in [M]$.
\end{itemize}

Equipped with the above random variables, we are ready to describe the graph constructions $\Gconn$ and $\GST$ that we use for connectivity and spanning tree construction, respectively.
It will turn out that both graphs are nearly identical except for two specific edges.

\begin{itemize} 
\item \textbf{Vertices of $\Gconn$ and $\GST$:} The vertices consist of the \emph{center vertices} $C=\set{c_1,\dots,c_M}$, the \emph{helper vertices} $S=\set{s_1,\dots,s_M}$, two \emph{hub vertices} $h_0$ and $h_1$, and the vertices in the set $W$ defined above.\footnote{By a slight abuse of notation, we use the same variables for the vertex sets as well as for their IDs.}
The IDs of the vertices in $C \cup S \cup \set{h_0,h_1}$ are fixed in advance and independent of the random variables sampled above, which means that they can be easily identified by their neighbors.
The number of described vertices is
\begin{align}
N_0 = M\cdot q + 2M+2q+2=q^3 +2q^2 +2q +2. \label{eq:N0}
\end{align}
In addition, the graph also contains $n - N_0$ many \emph{padding vertices}.
\item \textbf{Edges common to $\Gconn$ and $\GST$:}
\begin{align*}
  &\{c_i,s_i\} &&(i\in[M]),\\
  &\{c_i,p\}   &&(i\in[M],\ p\in T_i),\\
  &\{h_0,p\}   &&(i\in[M],\ p\in P_i\setminus T_i),\\
  &\{c_i,r\}   &&(i\in[M],\ r\in D_i^{X_i}),\\
  &\{h_0,r\}   &&(r\in R_0),\\
  &\{h_1,r\}   &&(r\in R_1).
\end{align*}
Note that every padding vertex is attached as a leaf to the hub $h_0$. 
\item \textbf{Special Edge of $\GST$:}
We directly connect the two hubs by adding $\set{h_0,h_1}$.
\item \textbf{Special Edge of $\Gconn$:}
We use the value of the random index $J \in [M]$ by adding the edge $\set{h_0,s_J}$.
\end{itemize}
We use $\Dconn$ to denote the distribution obtained by sampling the corresponding random variables and constructing the graph $\Gconn$.
Similarly, we use $\DST$ to denote the distribution for constructing $\GST$.

We conclude this section by showing some crucial properties of these graphs:

\begin{lemma} \label{lem:GST}
The following properties hold for a graph $\GST$ sampled from $\DST$:
\begin{compactenum} 
\item[(A)] Graph $\GST$ is connected.
\item[(B)] Consider two distinct binary vectors $X,X' \in \set{0,1}^M$. Then, any spanning tree of ${G}_{X}^{\text{ST}}$ contains an edge between a center node and $R_0 \cup R_1$ that does not exist in ${G}_{X'}^{\text{ST}}$. Consequently, ${G}_{X}^{\text{ST}}$ and ${G}_{X'}^{\text{ST}}$ do not have a common spanning tree.
\end{compactenum}
\end{lemma}
\begin{proof}
For (A), recall that the two hub vertices are connected by the edge $\set{h_0,h_1}$.
Every node in $P_i \setminus  T_i$ is adjacent to $h_0$, whereas every node in $R_a$ is connected to $h_a$ ($a \in \set{0,1}$). 
Moreover, every node in $T_i$ is connected to center $c_i$, as is the helper node $s_i$.
Finally, every $c_i$ has at least one neighbor in $R_0\cup R_1$ and thus has a path to either of the hubs.

We now prove (B). 
Suppose that $X_i \ne X_i'$, for some $i \in [M]$, and let $C_i=\set{c_i,s_i} \cup T_i$.
The only edges crossing the cut $(C_i,V(\GST)\setminus C_i)$ in $\GST$ are the ones between $c_i$ and nodes in $D_i^{X_i}  \subseteq R_{X_i}$. 
Thus, any given spanning tree of $\GST$ must contain an edge $e \in (C_i,V(\GST)\setminus C_i)$.
Similarly, the cut $(C_i,V({G}_{X'}^{\text{ST}})\setminus C_i)$ in ${G}_{X'}^{\text{ST}}$, only contains edges connecting $c_i$ to some neighbors in $D_i^{X_i'}  \subseteq R_{X_i'}$.
As above, any spanning tree contains an edge $e' \in (C_i,V({G}_{X'}^{\text{ST}})\setminus C_i)$.
Recall that $R_{X_i} \cap R_{X_i'} = \emptyset$, which implies that $e \notin {G}_{X'}^{\text{ST}}$. 
\end{proof}

\begin{lemma} \label{lem:Gconn}
The graph $\Gconn$ is connected if and only if $X_J=1$.
\end{lemma}
\begin{proof}
As in the proof of Lemma~\ref{lem:GST}, observe that
every node in $P_i \setminus  T_i$ is adjacent to $h_0$, whereas every node in $R_a$ is connected to $h_a$ ($a \in \set{0,1}$). 
Moreover, every node in $T_i$ is connected to center $c_i$, as is the helper node $s_i$.
Every $c_i$ has at least one neighbor in $R_0\cup R_1$ and thus has a path to either of the hubs.
In other words, every node has a path to one of the two hub nodes,
however, in contrast to $\GST$, there is no edge $\set{h_0,h_1}$ in $\Gconn$.

Now suppose that $X_J=1$ and recall that the edge $\set{h_0,s_J} \in E(\Gconn)$. 
Note that $D_J^1  \subseteq R_1$ is nonempty.
Thus, there exists a path $h_0 - s_J - c_J - r - h_1$, for some node $r \in R_1$, and thus $\Gconn$ is connected.
Finally, consider the case $X_J=0$. 
Let $\textsf{Pad}$ be the padding nodes (if any), and define
the sets 
\begin{align}
A &= \set{h_0} \cup R_0 \cup \bigcup_{i=1}^M (P_i \setminus  T_i) \cup \textsf{Pad} \cup \bigcup_{i:X_i=0}(\set{c_i,s_i} \cup T_i),\notag\\ 
B &= \set{h_1} \cup R_1 \cup \bigcup_{i:X_i=1}(\set{c_i,s_i} \cup T_i).
\end{align}
Observe that $A$ and $B$ are disjoint and partition the set of nodes.
According to the construction of $\Gconn$, the cut $(A,B)$ is empty, i.e., the graph is disconnected. 
\end{proof}
 
\section{Unsafe and Dangerous Splits} \label{sec:dangerous}

For each center node $c_i$ and each $A \subseteq W$, let $f_{i}(A)$ be the message sent by the assumed deterministic algorithm when center node $c_i$ has the neighborhood $A \cup \set{s_i}$.
That is, for the given list of neighbors, $f_{i}$ produces a message from an alphabet $\mathcal{M}$ of size at most $2^B$.
The alert reader may remark that $c_i$'s neighbors are chosen from $W \cup \set{s_i}$; however, we omit $s_i$ from the domain of $f_i$, because its ID is fixed and independent of the input distribution. 

\begin{lemma} \label{lem:all_unsafe}
Suppose we sample the partition of the IDs in $W$ according to Step~1A in Section~\ref{sec:dist}, and note that we can think of each triple $(P_i,R_0,R_1)$ as an OBS of the set of possible neighbors of $c_i$ in $W$.
Let $\allunsafe$ denote the event that, for every center $c_i$, $(P_i,R_0,R_1)$ is unsafe (see Def.~\ref{def:safe}).
\begin{align}
\Pr_{\text{\rm Step~1A}}\lt[  \allunsafe \rt] \ge \frac{1}{2}.\notag 
\end{align}
\end{lemma}
\begin{proof}
Consider a center $c_i$ and condition on the support  being the set $U_i=P_i \cup R_0 \cup R_1$, which has size $3q$.
According to the hard input distribution, the triple $(P_i,R_0,R_1)$  is uniform among all possible OBSs of $U_i$.
By Lemma~\ref{lem:OBS_safe}, it follows that $(P_i,R_0,R_1)$ is safe for $f_i$ with probability at most $2^{-q}$. A union bound over the $M=q^2$ centers shows that the event that there exists a center for which the sampled OBS is safe happens with probability at most $q^2 \cdot 2^{-q} \le \frac{1}{2}$, for sufficiently large $q$.
\end{proof}

\begin{definition} \label{def:dangerous}
Consider an unsafe OBS $(P_i,R_0,R_1)$ for $f_i$.
We say that the triple $(T_i, D_i^0, D_i^1)$ is \emph{dangerous} if 
\begin{align}
	f_i(T_i \cup D_i^0) = 	f_i(T_i \cup D_i^1).
\end{align}
\end{definition}

Intuitively speaking, we can think of a dangerous triple as being a ``witness'' to the fact that $(P_i,R_0,R_1)$ is unsafe.
We use $\allbad$ to denote the event that every sampled triple is dangerous, and derive a lower bound on its probability next: 

\begin{lemma}  \label{lem:all_bad}
Condition on sampling the partition of $W$ according to Step~1A of the input distribution in Section~\ref{sec:dist} and suppose that the OBS of every center node is unsafe.
That is,
\begin{align}
 \Pr_{\text{\rm Step 1B}}\lt[\ \allbad \ \md|\ \allunsafe \ \rt] \ge 2^{-3q^3}.  \notag 
\end{align}
\end{lemma}
\begin{proof}
Given $(P_i,R_0,R_1)$, the number of candidate triples of a center node is
$2^{q}(2^q-1)^2$, since there are $2^q$ choices for $T_i$ and $2^{q}-1$ choices each for the (nonempty) subsets $D_i^0$ and $D_i^1$.
Moreover, each of the 
\begin{align}
2^{q}(2^q-1)^2 \le 2^{3q}
\end{align} 
candidate triples is equally likely, according to the sampling performed in Step~1B of Section~\ref{sec:dist}.
Thus, we choose a dangerous triple for $c_i$ with probability at least $2^{-3q}$.
Recalling that the choices for distinct center nodes are conditionally independent, it follows that
we sample dangerous triples for all of the $M$ centers with probability at least $2^{-3qM}=2^{-3q^3}$.
\end{proof}

\begin{definition}[Right-side Transcript] \label{def:rside}
Let $N(r)$ denote the IDs of the neighbors of a node $r$ and let $f_r$ be the (deterministic) message function used by the algorithm at $r$.
We define the \emph{right-side transcript} to be the sequence
\begin{align}
	Z = (f_r(N(r)))_{r \in R_0 \cup R_1},
\end{align}
where we fix the order of the nodes in $R_0 \cup R_1$ according to their IDs.
\end{definition}

\begin{lemma} \label{lem:allsame}
Condition on events $\allbad$, $\allunsafe$ and the sampling in Steps~1A and 1B.
Then, for any two distinct binary vectors $X,X' \in \set{0,1}^M$, it holds that every center node sends the same message for graphs $G_X^{\text{ST}}$ and $G_{X'}^{\text{ST}}$. 
Moreover, the same statement also holds for $\Gconn$ and $G_{X',J}^{\text{conn}}$.
\end{lemma}
\begin{proof}
Under the given conditioning, we have already performed the sampling of Steps~1A and 1B, thus we only need to consider the possible impact of the choices made in Steps~2 and 3. 
According to the construction of $\GST$ (see Sec.~\ref{sec:dist}), the neighborhood of a center node $c_i$ is independent of sampling the index $J$ in Step~3.
Moreover, $c_i$ is connected to the same helper node $s_i$ in all graphs sampled from $\DST$.
Thus, it follows from Definition~\ref{def:dangerous}, that $c_i$ sends the same message no matter what vector $X$ is sampled in Step~2.

To complete the proof, note that all properties used in the proof also hold for the graph construction $\Gconn$.
\end{proof}

\begin{lemma} \label{lem:rside}
Condition on the partition and all triples obtained in Steps~1A and 1B of the sampling process in Section~\ref{sec:dist}, and suppose that events $\allbad$ and $\allunsafe$ occur.
\begin{compactenum} 
\item[(A)] The only messages depending on the binary vector $X$ (sampled in Step~2) are part of the right-side transcript $Z$, which itself is fully determined by $X$ under the given conditioning, i.e., $Z=Z(X)$.
\item[(B)] There are at most $2^{2qB}$ many choices for $Z$.
\end{compactenum} 
\end{lemma}
\begin{proof} 
Part (B) follows by observing that there are $2q$ vertices in $R_0 \cup R_1$ and each of them sends a message of at most $B$ bits.

We now focus on Part~(A).
It is immediate from the construction of both, $\GST$ and $\Gconn$, that only the messages sent by the center nodes and the nodes in $R_0 \cup R_1$ may depend on $X$. 
Furthermore, the conditioning on $\allbad$ guarantees that all center nodes send the same message for every choice of $X$, as shown in Lemma~\ref{lem:allsame}, and hence the only nodes that can distinguish $X$ from $X'$ lie in $R_0 \cup R_1$.
\end{proof}

\section{The Distributional Error of Spanning Tree Construction} \label{sec:distST}

Let $\fail$ denote the event that the referee outputs an incorrect answer.

\begin{lemma} \label{lem:ST_fails_X}
Condition on the partition and all triples obtained in Steps~1A and 1B, and suppose that events $\allbad$ and $\allunsafe$ occur.
Then,
\begin{align}
\Pr_X\lt[\ \fail \ \md|\ \allunsafe,\allbad \ \rt] \ge \frac{1}{2}.\notag 
\end{align}
\end{lemma}
\begin{proof}
Let $\mathcal{S}$ be the set of binary vectors $X$ such that the algorithm correctly outputs a spanning tree for $\GST$.
As a first step, we show that the right-side transcript (see Def.~\ref{def:rside}) $Z(X)$ has distinct values for all $X \in \mathcal{S}$.
To this end, recall from Lemma~\ref{lem:GST}(B) that the algorithm must output distinct spanning trees for all distinct $X$ and $X'$.
By Lemma~\ref{lem:rside}(A), we know that only the messages in $Z$ depend on $X$.
It follows that $Z(X) \ne Z(X')$, as otherwise the referee would output the same spanning tree for $X$ and $X'$. 

Since $Z$ is injective on $\mathcal{S}$, Lemma~\ref{lem:rside}(B) tells us that 
\[
|\mathcal{S}| \le 2^{2qB}.
\]
Given that $X$ is uniform on $\set{0,1}^M$ and $M=q^2$, it follows that the algorithm succeeds with probability at most
\begin{align}
 2^{2qB - M} \le 2^{-49M/50} \le \frac{1}{2},
\end{align}
where we used the assumption that $q$ is sufficiently large.
\end{proof}

\begin{lemma} \label{lem:distST}
Consider a deterministic algorithm in which each node sends a $B$-bit message, where $q$ is any sufficiently large integer.
If $n \ge q^3 + 2q^2 +2q +2$ and $B \le \frac{q}{100}$, then
the distributional error of computing a spanning tree on $\DST$ is at least $\Omega\lt( 2^{-3q^3} \rt)$.
\end{lemma}
\begin{proof} 
We have
\begin{align}
\Pr_{\DST}\lt[  \fail\ \rt] 
&\ge 
\Pr_{\DST}\lt[  \fail\ \md|\ \allbad,\allunsafe \rt] \cdot 
\Pr_{\DST}\lt[  \allbad\ \md|\ \allunsafe \rt] \cdot 
\Pr_{\DST}\lt[  \allunsafe\  \rt].\notag
\intertext{
	Conditioned on the partition and triples sampled in Steps~1A and 1B as well as on events $\allbad$, $\allunsafe$, the remaining random choices are the sampling of $X$ and $J$ in Steps~2 and 3, respectively, and $\GST$ is conditionally independent of $J$.  
  We know from Lemma~\ref{lem:GST}(B) that correct outputs for distinct vectors $X$ and $X'$ result in distinct spanning trees, and Lemma~\ref{lem:rside}(A) tells us that only the right-side transcript, i.e., the messages sent by nodes in $R_0 \cup R_1$ can depend on the binary vector at hand.
$\Pr_{\DST}\lt[  \fail\ \md|\ \allbad,\allunsafe \rt]
=\EE\lt[ \Pr_{X}\lt[  \fail\ \md|\ \allbad,\allunsafe \rt] \rt]
\ge \frac{1}{2}$. %
Plugging this bound into the right-hand side of the above inequality, we obtain
}
\Pr_{\DST}\lt[  \fail\ \rt] 
&\ge 
\frac{1}{2}\cdot 
\Pr_{\text{\rm Step 1B}}\lt[  \allbad\ \md|\ \allunsafe \rt] \cdot 
\Pr_{\text{\rm Step 1A}}\lt[  \allunsafe  \rt] \notag \\
\ann{by Lem.~\ref{lem:all_unsafe} and \ref{lem:all_bad}}
&\ge
\frac{1}{2} \cdot 2^{-3q^3} \cdot \frac{1}{2} \notag \\
&= \Omega\lt(  2^{-3q^3}  \rt).\notag 
\end{align}
\end{proof}

\section{The Distributional Error of Connectivity Testing} \label{sec:distconn}

For each $j \in [M]$, we create a \emph{deterministic decoder $g_j$} by hardwiring $J=j$ and all fixed messages that are not part of the right-side transcript. 
Thus, 
\begin{align}
g_j : \text{range(Z)} \to \set{0,1},
\end{align}
whereby the output of $g_j(Z(X))$ is the assumed graph sketching algorithm's decision regarding whether $\Gconn$ is connected. 
Recall from Lemma~\ref{lem:Gconn} that the correct answer is $X_J$.
We make use of the following standard random-access inequality, whose proof is similar to Fano's inequality and the concavity of the binary entropy function. We include a complete proof in Appendix~\ref{app:random_access} for completeness.

\newcommand{\lemRandomAccess}{
Let $X \in \set{0,1}^M$ be chosen uniformly at random.
Moreover, let $Z=Z(X)$ be a deterministic encoding and let $g_1,\dots,g_M$ be deterministic binary decoders.
If $J$ is uniform on $[M]$ and independent of $X$, then
\begin{align}
\log_2 |\text{range}(Z)| \ge M (1 - h_2(\epsilon)),\notag 
\end{align}
where $\epsilon = \Pr\lt[ g_J(Z) \ne X_J \rt] $ and $h_2$ is the binary entropy function.
}
\begin{lemma} \label{lem:random_access}
\lemRandomAccess
\end{lemma}

\begin{lemma} \label{lem:XJconn}
Condition on the partition and all triples obtained in Steps~1A and 1B, as well as on events $\allbad$ and $\allunsafe$.
Then,
\begin{align}
\Pr_{X,J}\lt[ \fail  \ \md|\ \allunsafe, \allbad \rt] \ge \frac{1}{4}.
\notag 
\end{align}
\end{lemma}
\begin{proof}
Assume towards a contradiction that $\Pr_{X,J}\lt[ \fail  \ \md|\ \allunsafe, \allbad \rt] < \frac{1}{4}.$
By Lemma~\ref{lem:random_access} and the fact that $h_2$ is strictly increasing on $[0,1/2]$, we have
\begin{align}
\log_2|\text{range}(Z)| > M(1 - h_2(1/4)) > \frac{M}{8}, \label{eq:r1}
\end{align}
where the last inequality follows because $h_2(1/4) < \frac{7}{8}$.
Recalling that $B \le q/100$ and $M = q^2$, and applying Lemma~\ref{lem:rside}(B) yields
\begin{align}
\log_2|\text{range}(Z)| \le 2qB \le \frac{M}{50},
\end{align}
contradicting \eqref{eq:r1}.
\end{proof}

\begin{lemma} \label{lem:distConn}
Consider a deterministic algorithm in which each node sends a $B$-bit message, where $q$ is any sufficiently large integer.
If $n \ge q^3 + 2q^2 +2q +2$ and $B \le \frac{q}{100}$, then
the distributional error of deciding connectivity on graphs sampled from $\Dconn$ is at least $\Omega\lt( 2^{-3q^3} \rt)$.
\end{lemma}
\begin{proof} 
The proof is similar to Lemma~\ref{lem:distST} with the key difference being that we apply Lemma~\ref{lem:XJconn} instead of Lemma~\ref{lem:ST_fails_X}.

We have
\begin{align}
\Pr_{\Dconn}\lt[  \fail\ \rt] 
&\ge 
\Pr_{\Dconn}\lt[  \fail\ \md|\ \allbad,\allunsafe \rt] \cdot 
\Pr_{\Dconn}\lt[  \allbad\ \md|\ \allunsafe \rt] \cdot 
\Pr_{\Dconn}\lt[  \allunsafe\  \rt].\notag
\intertext{
	Conditioned on the partition and triples sampled in Steps~1A and 1B as well as on events $\allbad$, $\allunsafe$, the remaining random choices are the sampling of $X$ and $J$ in Steps~2 and 3.  
Thus, by Lemma~\ref{lem:XJconn},
$\Pr_{\Dconn}\lt[  \fail\ \md|\ \allbad,\allunsafe \rt]
=\EE\lt[ \Pr_{X,J}\lt[  \fail\ \md|\ \allbad,\allunsafe \rt] \rt]
\ge \frac{1}{4}$. %
Plugging this bound into the right-hand side of the above inequality, we obtain
}
\Pr_{\Dconn}\lt[  \fail\ \rt] 
&\ge 
\frac{1}{4}\cdot 
\Pr_{\text{\rm Step 1B}}\lt[  \allbad\ \md|\ \allunsafe \rt] \cdot 
\Pr_{\text{\rm Step 1A}}\lt[  \allunsafe  \rt] \notag \\
\ann{by Lem.~\ref{lem:all_unsafe} and \ref{lem:all_bad}}
&\ge
\frac{1}{4} \cdot 2^{-3q^3} \cdot \frac{1}{2} \notag \\
&= \Omega\lt(  2^{-3q^3}  \rt).\notag 
\end{align}
\end{proof}

\section{Completing the Proof of Theorem~\ref{thm:main}} \label{sec:random}

\begin{reptheorem}{thm:main}
\thmMain
\end{reptheorem}
\begin{proof} 
Given a randomized algorithm with a worst case message length of a certain number of bits, it is straightforward to modify the algorithm such that every message has length exactly $B$ without changing the asymptotic size of the message alphabet.
Thus, we assume that nodes always send $B$-bit messages throughout the rest of the proof, which means that the size of the message alphabet is at most $2^B$.

We first consider the connectivity problem.
Let $n' = \min\Set{n,\log_2\frac{1}{\delta}}$, and choose 
\begin{align}
	q = \lt\lfloor \lt( \frac{n'}{64} \rt)^{1/3} \rt\rfloor,
\end{align}
which implies 
\begin{align}
  2^{-3q^3} \ge 2^{-3n'/64}. \label{eq:i1}
\end{align}
For sufficiently large $n'$, it follows from \eqref{eq:N0} that 
\begin{align}
N_0 \le \frac{n'}{16} \le n,\notag 
\end{align}
which means that the hard distributions $\Dconn$ and $\DST$ (see Sec.~\ref{sec:dist}) are supported on $n$-node graphs, possibly by including $n-N_0$ padding vertices. 

Now, assume towards a contradiction that $B \le \frac{q}{100}$.
Combining the distributional error guaranteed by Lemma~\ref{lem:distConn}  with Yao's minimax lemma~\cite{DBLP:conf/focs/Yao77}, reveals that  
\begin{align}
\Pr\lt[ \fail \rt] &\ge \Omega\lt( 2^{-3q^3} \rt) \notag\\
\ann{by \eqref{eq:i1}}
&\ge \Omega\lt( 2^{-{3n'}/{64}} \rt) 
\ge \Omega\lt( \delta^{3/64} \rt) = \omega\lt( \delta \rt),
\end{align}
contradicting the assumed error probability being at most $\delta$.
It follows that $B = \Omega\lt( q \rt)$, which completes the proof for connectivity.
The proof of the result for spanning tree construction is analogous: We combine Lemma~\ref{lem:distST} (instead of Lemma~\ref{lem:distConn}) with Yao's lemma and note that the graph $\GST$ is guaranteed to be connected by Lemma~\ref{lem:GST}.
\end{proof}

\section{Extension to $k$-Edge Connectivity} \label{sec:kedge}

In this section, we prove Corollary~\ref{cor:kedge} via a reduction from the connectivity lower bound obtained in Theorem~\ref{thm:main}.
To prove a lower bound for $k$-edge connectivity on an $n$-node graph, we define $H$ to be a graph of 
\begin{align}
	m = \lt\lfloor \frac{n}{k+1} \rt\rfloor 
\end{align}
nodes $u_1,\dots,u_{m}$.
A straightforward counting argument shows that there is a partition of $[n] = V_1 \cup .\dots \cup V_m$ such that 
\begin{align}
k+1 \le |V_i| \le 2(k+1) \qquad (i \in [m]). \label{eq:csize}
\end{align}

\begin{definition} \label{def:blowup}
We define the \emph{blow-up} $B_k(H)$ on $[n]$ as follows:
 \begin{compactitem} 
\item We replace each $u_i$ with a clique among the vertices in $V_i$, and call $V_i$ a \emph{cluster} in $B_k(H)$.
\item For every edge $\set{u_i,u_j} \in E(H)$, we add a complete bipartite graph on the vertex sets $V_i$ and $V_j$. 
\end{compactitem}
\end{definition}

\begin{lemma} \label{lem:blowup_conn}
$H$ is connected if and only if $B_k(H)$ is $k$-edge connected.
\end{lemma}
\begin{proof} 
Observe that, if $H$ is disconnected, and consider two vertices $u_i$ and $u_j$ lying in different connected components.
According to Def.~\ref{def:blowup}, we do not add any edges between the corresponding clusters $V_i, V_j  \subseteq V(B_k(H))$, and hence $B_k(H)$ is disconnected as well.

Now suppose that $H$ is connected. 
Consider any nontrivial cut $S$ of $B_k(H)$:
First, suppose that $S$ splits a cluster $V_i$. Let $a = |S \cap V_i|$. The clique edges of $V_i$ contribute at least
\begin{align}
a(|V_i| - a) \ge |V_i| - 1 \ge k,
\notag 
\end{align}
where the latter inequality follows from \eqref{eq:csize}.

Otherwise, it must be that $S$ is a union of complete clusters $V_1,\dots V_\ell$. 
Since $H$ is connected, it follows that there is at least one edge $\set{u_j,u_{j'}}$ in the cut $\lt( \bigcup_{j=1}^{\ell} u_j, V(H) \setminus  \bigcup_{j=1}^{\ell} u_j \rt)$.
According to the construction of $B_k(H)$, the vertices in $V_j$ and $V_{j'}$ form a complete bipartite graph.
By \eqref{eq:csize}, we have 
\begin{align}
|V_j||V_{j'}| \ge (k+1)^2 \ge k
\end{align}
edges across the cut.
\end{proof}

The simulation of $B_k(H)$ is straightforward since each node knows the fixed partitioning of $[n]$ and can use its local neighborhood information to determine the edges incident to its simulated nodes in $B_k(H)$.
The node $u_i$ then simply concatenates the simulated messages and sends them to the referee, which, due to \eqref{eq:csize}, may  require a message length of up to 
$
|V_i|B \le 2(k+1)B
$
bits.
Applying Theorem~\ref{thm:main} to the simulated connectivity protocol, we get $
2(k+1)B \ge \Omega\lt( \min\Set{m,\log_2\frac{1}{\delta}}^{1/3} \rt)$, and therefore
\begin{align}
B &\ge \Omega\lt( \frac{1}{k} \cdot \min\Set{m,\log_2\frac{1}{\delta}}^{1/3}  \rt)\notag\\ 
&\ge \Omega\lt( \frac{1}{k} \cdot \min\Set{\frac{n}{k},\log_2\frac{1}{\delta}}^{1/3}  \rt),  \label{eq:lb2}
\end{align}
where the second inequality follows due to $m \ge \frac{n}{2(k+1)}$ (see \eqref{eq:csize}).

Now consider the exponentially small error regime, i.e., $\delta \le 2^{-\alpha n}$ and thus
\begin{align}
\log_2 \frac{1}{\delta} \ge \alpha n \ge \alpha m,\notag
\end{align}
which implies
\begin{align}
\min\Set{m,\log_2 \frac{1}{\delta}} \ge
m  \cdot \min\Set{1,\alpha }.\notag
\end{align}
Recalling \eqref{eq:lb2}, it follows that
\begin{align}
B = \Omega\lt( \min\Set{1,\alpha}^{1/3} \frac{n^{1/3}}{k^{4/3}} \rt)
= \Omega_\alpha\lt( \frac{n^{1/3}}{k^{4/3}} \rt). 
\end{align}

\section{AI Disclosure} \label{sec:AI}
We used GPT 5.5 as part of our research.
In particular, GPT 5.5 suggested to use the theorem of \cite{karpovsky1978coordinate} for bounding the size of a certain family  of the ordered balanced splits with a given 3-ary VC dimension, and also suggested the ``blow-up'' simulation for extending the results to $k$-edge connectivity. 
However, no AI-generated text was included in the paper. 
All content was written and verified by the authors, including all statements of lemmas, theorems, and their proofs. 
 The authors take full responsibility for all content.

\appendix
\section*{APPENDIX}
\section{Proof of Lemma~\ref{lem:random_access}} \label{app:random_access}
\begin{replemma}{lem:random_access}
\lemRandomAccess
\end{replemma}
\begin{proof}  
For each $i$, define
\begin{align}
  \epsilon_i&=\Pr[g_i(Z)\ne X_i],\notag\\ 
  E_i&=X_i\oplus g_i(Z).
\end{align}
Given $Z$, the variables $X_i$ and $E_i$ determine one another. Thus,
\begin{align}
  H(X_i\mid Z)=H(E_i\mid Z)\le H(E_i)=h_2(\epsilon_i).\notag 
\end{align}
By subadditivity,
\begin{align*}
  H(X\mid Z)
  &\le\sum_{i=1}^M H(X_i\mid Z)\\
  &\le\sum_{i=1}^M h_2(\epsilon_i)\\ \ann{by concavity of $h_2$}
  &\le M h_2\!\left(\frac1M\sum_{i=1}^M\epsilon_i\right)
   =M h_2(\epsilon).
\end{align*}
Since $X$ is uniform, $H(X)=M$.  Therefore
\[
  I(X:Z)=H(X)-H(X\mid Z)\ge M(1-h_2(\epsilon)).
\]
We conclude that
\[
  I(X:Z)\le H(Z)\le\log_2|\range(Z)|.
\]
\end{proof}

\bibliographystyle{alpha}
\bibliography{refs}

\end{document}